\documentclass[12pt, draftclsnofoot, onecolumn]{IEEEtran}
\usepackage[utf8]{inputenc}
\usepackage[english]{babel}
\usepackage{cite}
\usepackage{setspace}
\usepackage{graphicx}
\usepackage{subcaption}
\usepackage{amsmath}
\usepackage{algorithm}
\usepackage{algcompatible}
\usepackage{array}
\usepackage{amsthm}
\usepackage{fixltx2e}
\usepackage{txfonts}
\usepackage{stfloats}
\newtheorem{theorem}{\textbf{Theorem}} 
 
\newtheorem{lemma}{\textbf{Lemma}}

\newtheorem{definition}{\textbf{Definition}}

\usepackage{lipsum}  
\usepackage{color}

\begin{document}
\title{Minimum Length Scheduling for Discrete-Rate Full-Duplex Wireless Powered Communication Networks
}
\author{Muhammad Shahid~Iqbal, Yalcin Sadi,~\IEEEmembership{Member,~IEEE}
        and~Sinem Coleri,~\IEEEmembership{Senior Member,~IEEE}
\thanks{Muhammad Shahid~Iqbal and Sinem Coleri are with the department of Electrical and Electronics Engineering, Koc University, Istanbul, Turkey, email: $\lbrace$miqbal16, scoleri$\rbrace$@ku.edu.tr. Yalcin Sadi is with the department of Electrical and Electronics Engineering, Kadir Has University, Istanbul, Turkey, email: yalcin.sadi@khas.edu.tr.

This work is supported by Scientific and Technological Research Council of Turkey Grant $\#$117E241. 

This work is partially presented at the $18^{th}$ International Conference on Ad Hoc Networks and Wireless (AdHoc-Now 2019).
}}
\maketitle
\begin{abstract}
In this paper, we consider a wireless powered communication network where multiple users with RF energy harvesting capabilities communicate to a hybrid energy and information access point (HAP) in full-duplex mode. Each user has to transmit a certain amount of data with a transmission rate from a finite set of discrete rate levels, using the energy initially available in its battery and the energy it can harvest until the end of its transmission.  Considering this model, we propose a novel discrete rate based minimum length scheduling problem to determine the optimal power control, rate adaptation and transmission schedule subject to data, energy causality and maximum transmit power constraints. The proposed optimization problem is proven to be NP-hard which requires exponential-time algorithms to solve for the global optimum. As a solution strategy, first, we demonstrate that the power control and rate adaptation, and scheduling problems can be solved separately in the optimal solution. For the power control and rate adaptation problem, we derive the optimal solution based on the proposed minimum length scheduling slot definition. For the scheduling, we classify the problem based on the distribution of minimum length scheduling slots of the users over time. For the non-overlapping slots scenario, we present the optimal scheduling algorithm. For the overlapping scenario, we propose a polynomial-time heuristic scheduling algorithm.
\end{abstract}
\IEEEpeerreviewmaketitle
\section{Introduction} \label{sec:intro}
Due to easy maintenance, low installation cost and flexibility, time critical wireless sensor networks are being rapidly adopted in cyber-physical and emergency alert systems \cite{wncs_sinem, intravehicle_sinem}. To meet the traffic demand of such battery limited networks, minimum length scheduling is very important but remained under cover except few studies \cite{wncs_ref83, wncs_ref194}. The recent advances in the energy harvesting technologies have offered a great potential to overcome the battery replacement and to provide a perpetual energy. Natural resources such as sun, wind and vibration are the possible sources of energy harvesting \cite{natural_1, natural_2} but their randomness and environmental dependence is the bottleneck for their implementation in such networks. Similarly, inductive or magnetic resonant coupling technologies are also infeasible for such networks due to large size, calibration issues and short energy transfer range. Radio Frequency (RF) based energy harvesting technology is capable to master these hurdles with its high range, small form factor and full control on energy transfer \cite{harvest_10}. In RF energy harvesting, radio signals with frequency range $3$ kHz-$300$ GHz are used as a medium to carry energy in the form of electromagnetic radiation. 

Basically, there are two operational models for RF energy harvesting known as simultaneous wireless information and power transfer (SWIPT) and wireless powered communication networks (WPCN). In SWIPT, the same RF signal is used for both wireless information transfer and wireless power transfer. For a single user SWIPT, the trade-off between information transmission rate and energy transfer efficiency have been studied for additive white Gaussian noise channel \cite{harvest_06}, flat fading channel \cite{harvest_61}, and a non linear energy harvesting model \cite{harvest_new63}. The multi-user SWIPT networks are studied, to minimize the base station transmit power \cite{harvest_59} or to maximize the energy transfer under a minimum signal to noise ratio (SNR) constraint  \cite{harvest_19, harvest_60}. The data buffer constrained throughput maximization is studied in \cite{harvest_new62}. The scheduling mechanism in these SWIPT based models is missing except \cite{harvest_01_ref146} where the scheduling is considered in a different context i.e. a downlink multi-user scheduling for time slotted system is proposed in which the HAP have to transmit data to all the users in the downlink and the information recipient user is determined for each time slot. The scheduling algorithm is proposed for the selection of this single information recipient user in each time slot. 
 
In WPCN, ``Harvest-then-transmit" is the introductory protocol, where all the users first harvest the RF energy broadcost by the hybrid access point (HAP) in the downlink and then transmit their information to the HAP in the uplink by using time division multiple access \cite{harvest_07}. The objective of this study is to maximize the sum throughput by jointly optimizing the time allocation under a total time constraint. This sum throughput maximization results in an unfair time allocation among the users by supporting the near users due to their favouring channel conditions and low distance. For example, the user with maximum SNR can occupy the total time which will result in a maximum throughput but the users with high path loss will be struggling to get access to the channel. To overcome this unfair resource allocation some studies have focused on the alternative objective functions such as maximization of weighted sum rate \cite{harvest_44, harvest_55}, total effective throughput maximization \cite{effective_throughput}, minimum throughput maximization \cite{harvest_04}, and energy efficiency maximization \cite{harvest_51,Energy_efficiency_Max}. In such half duplex systems, all the users have equal energy harvesting time therefore, the transmission order is not important and scheduling is not required for such system models. Although, in these studies the HAP is either restricted to transmit a constant power \cite{harvest_19, harvest_07, harvest_39_ref17,harvest_04,harvest_56} or its maximum and average transmit power is constrained by a limit \cite{harvest_50, harvest_44, harvest_55,SecureWPCN}, no maximum transmit power constraint for the users is considered. Furthermore, these studies consider the continuous data rate which is hard to realize practically. Also, the initial battery level of the users is not considered in these works, except \cite{harvest_51,LongTermThroughput}.

Recently, the WPCN has started to incorporate the full duplex technology to achieve high spectral and energy transfer efficiency by allowing the HAP to transmit energy and receive information simultaneously in the same frequency band. Self-Interference is the major setback for such co-time and co-frequency transmissions in which a part of the transmitted signal is received by itself, thus causing an interference to the received signal. Thanks to the recent self-interference cancellation (SIC) techniques \cite{harvest_41_ref26, harvest_41_ref27} and their practical implementations \cite{harvest_30_ref19, harvest_30_ref26}, full duplex has become the key transceiving technique for the 5G networks \cite{harvest_new66}. The main objective of the full duplex WPCN is to maximize the sum throughput by assuming either only the HAP is operating in full duplex mode \cite{harvest_30, harvest_40, harvest_50} or the HAP and users both are operating in the full duplex mode \cite{harvest_41, harvest_new68} under the perfect SIC environment. The studies in \cite{harvest_40, harvest_41, harvest_50} consider the residual self interference proportional to the HAP transmit power. Only \cite{harvest_30} considers the schedule length minimization given the traffic demand of the users. In full duplex systems, since the users are capable to harvest the energy during their own and other users data transmission time, so all the users have different energy harvesting time. This uneven energy harvesting time necessitates an efficient user ordering to minimize the scheduling length. However, the foregoing studies assume the pre-determined transmission order without incorporating any scheduling algorithm for an impracticable continuous data rate assumption. Furthermore, none of the mentioned studies have considered maximum user transmit power limitation, while assuming a constant HAP power \cite{harvest_30, harvest_40} or a maximum HAP power constraint \cite{harvest_50, harvest_41}. Moreover, these studies have assumed that all the harvested power should be used in the same transmission frame without considering the initial battery level of the users.      

The goal of this paper is to determine optimal rate adaptation, power control and time allocation, and transmission order with the objective to minimize the schedule length for a realistic in-band full-duplex WPCN model. The original contributions of this paper are listed below:

\begin{itemize}
\item We propose a new discrete rate WPCN optimization framework for a realistically-modeled in-band full-duplex energy harvesting network, employing the energy causality constraint and maximum user transmit power constraint which considers the initial battery levels and the energy storage capability of the users based on the discrete-rate transmission model.
\item We characterize the Discrete Rate Minimum Length Scheduling Problem ($\cal{DR-MLSP}$) to ascertain the optimal rate adaptation, power control, and schedule which minimizes the transmission completion time of all the users subject to traffic demand, maximum user transmit power and energy causality constraints based on the proposed discrete rate WPCN model. We formulate the problem mathematically as a mixed integer non-linear programming (MINLP) problem and proved that the problem is NP-Hard.
\item We formulate the power control and rate adaptation problem for a pre-determined transmission order. The optimal solution is presented based on the proposed minimum length scheduling slot definition.
\item For the scheduling problem, we define the minimum length scheduling (MLS) slot as the slot which yields minimum transmission completion time while starting transmission at any time instant $s_i$ for each user. We classify the $\cal{DR-MLSP}$ problem based on the distribution of the minimum length scheduling slots of users over time into non-overlapping and overlapping scenarios. For the non-overlapping scenario, we propose optimal scheduling algorithm. For the overlapping scenario, we propose a polynomial-time heuristic scheduling algorithm.
\item We evaluate the performance of the proposed scheduling schemes for different network scenarios and parameters including different minimum signal to noise ratio levels of the users, different HAP transmit powers and different network sizes in comparison to the optimality and pre-determined order transmission schemes. We illustrate that the proposed polynomial time heuristic algorithm perform very close to optimal while outperforming the predetermined schedule.
\end{itemize}

The rest of the paper is organized as follows. In Section \ref{sec:system}, we describe the discrete rate WPCN system model and assumptions used throughout the paper. In Section \ref{sec:dr-mlsp}, we present the mathematical formulation of the discrete rate minimum length scheduling problem and the proof of NP-Hardness for the proposed problem. In Section \ref{sec:power}, we present the optimal power control and rate adaptation problem for a pre-determined transmission order and derive its optimal solution. In Section \ref{sec:scheduling}, we investigate the optimal scheduling problem. In Section \ref{sec:simulation}, we evaluate the performance of the proposed scheduling schemes. Finally, concluding remarks are presented in Section \ref{sec:conclusion}.

\section{System Model and Assumptions} \label{sec:system}

The system model and assumptions are described as follows:

 \begin{figure}[t]
 \centering
\includegraphics[width= 0.5 \linewidth]{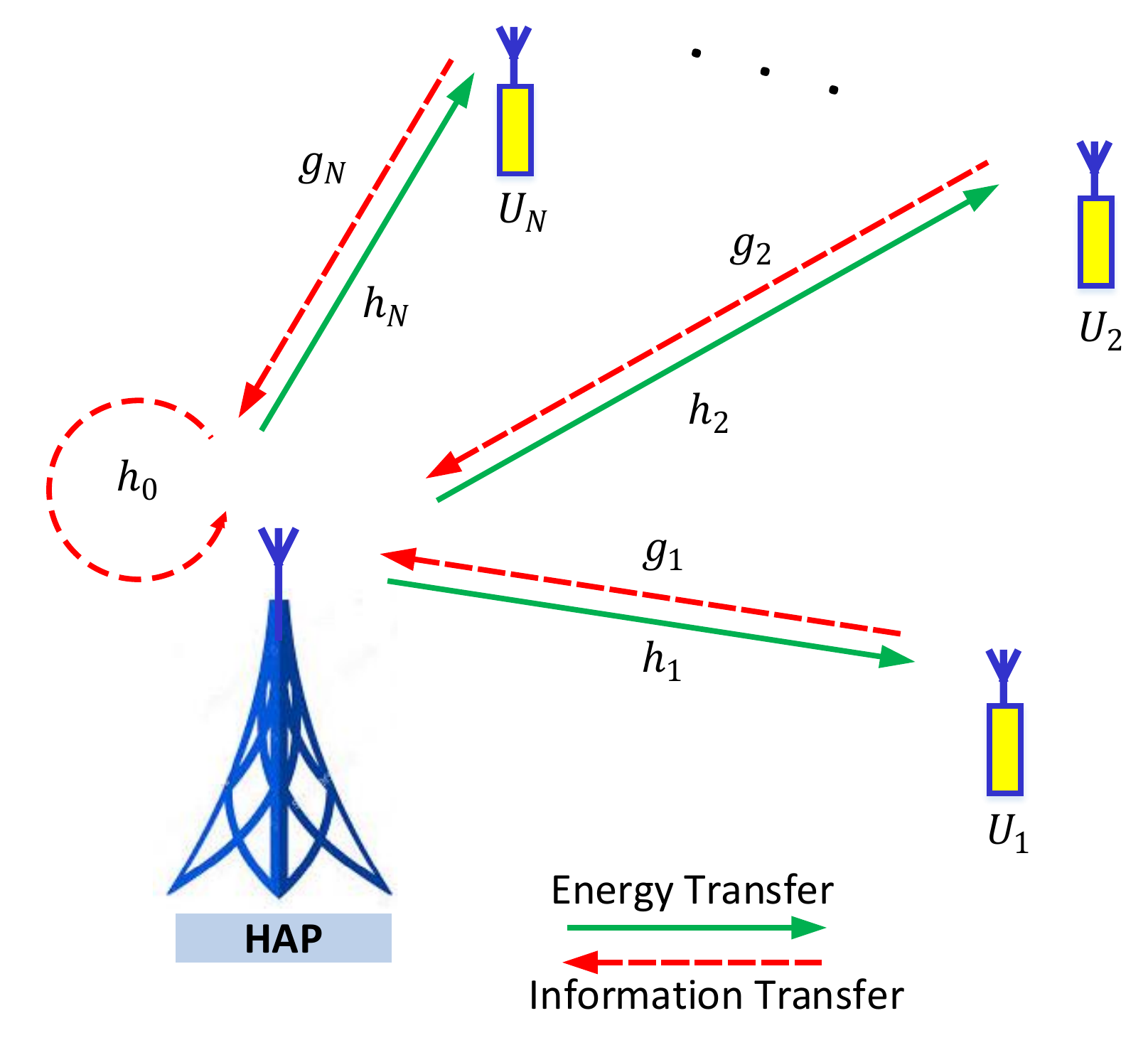}
\caption{System Model for Wireless Powered Communication Network} \label{Figure:1}
\end{figure}
\begin{enumerate} 
\item The WPCN architecture consists of a HAP and N users i.e., sensors or machine type communication devices as shown in Fig. \ref{Figure:1}. The HAP and all the N users are equipped with a full-duplex antenna which will be used to transfer the data on the uplink channel, and to transfer energy on the downlink channel simultaneously.   
\item In the proposed architecture, the uplink and downlink channel gains are assumed to be different. The uplink channel gain from user $i$ to the HAP is denoted by $g_i$ and the downlink channel gain from the HAP to user $i$ is denoted by $h_i$. Additionally, both the channels are assumed to be block-fadding channels i.e., the channel gains remain unchanged over the scheduling frame \cite{harvest_19, harvest_04, harvest_39_ref17, harvest_39, harvest_51, harvest_56}. We also assume that the HAP has perfect channel state information i.e., the channel gains are perfectly known at the HAP \cite{harvest_07,harvest_04,harvest_50, harvest_44, harvest_55,harvest_51,harvest_30,harvest_40}. 
\item The HAP is connected to a fixed power supply, and it will transmit a constant power $P_h$. In contrast to the HAP, the users have no other external power supply i.e., they are totally dependent on the harvested energy. All the users will harvest the energy transmitted by the HAP and they will store it in a rechargeable battery. The initial battery level of user $i$ is denoted by $B_i$ which contains the harvested and unused energy during the previous transmission cycles.

\item Time division multiple access protocol is used as the medium access protocol for the uplink data transmission from the users to the HAP. The whole time, during which the system will remain operational is partitioned into variable length scheduling frames. Each frame is further divided into variable time slots for user allocation. The HAP transmits continuous power throughout the scheduling frame. Each user will transmit its uplink data in the allocated time slot. In addition to stored energy, the energy harvested until the end of transmission will also be used for the data transmission. The energy harvested after the transmission will be stored for future usage.   
\item We assume a realistic non-linear energy harvesting model based on logistic function \cite{NLEH_model_01,NLEH_model_02,NLEH_01} which performs close to the experimental results proposed in \cite{NLEH_Prac_01,NLEH_Prac_02,NLEH_Prac_03}. In this non-linear energy harvesting model the energy harvesting rate for user $i$ is given as: 
\begin{equation}
C_i=\dfrac{P_s[\Psi_i-\Omega_i]}{1-\Omega_i}
\end{equation}
Where, $\Omega_i=\dfrac{1}{1+e^{A_iB_i}}$ is a constant to make sure zero-input zero-output response, $P_s$ is the maximum harvested power during saturation and $\Psi_i$ is the logistic function related to user $i$ and is given by:
\begin{equation}
\Psi_i=\dfrac{1}{1+e^{-A_i(h_iP_h-B)}}
\end{equation}
Where, $A$ and $B$ are the positive constants related to the non-linear charging rate with respect to the input power and turn-on threshold respectively. For a given energy harvesting circuit, the parameters $P_s$, $A$ and $B$ can be determined by curve fitting. 
\item We assume user $i$ has a traffic demand $D_i$ bits to be transmitted over the scheduling frame. 
\item We use discrete rate transmission model, in which a finite set of rates $r=(r^1, r^2,\cdots, r^M)$ and a finite set of SINR levels $\gamma=(\gamma^1, \gamma^2,\cdots, \gamma^M)$ are determined such that user $i$ can transmit at rate $r^k$ in the allocated time slot if the SINR achieved for user $i$ is:
 \begin{equation} \label{transmission_rate}
 \gamma_i = \dfrac{P_ig_i}{N_oW+\beta P_h}\geq \gamma^k,
 \end{equation}
where, $N_o$ is the noise density and $W$ is the bandwidth. The term $\beta P_h$ is the power of self interference at the HAP.
\item We use continuous power model in which the transmission power of a user can take any value below a maximum level $P_{max}$, which is imposed to avoid the interference to nearby systems.
\end{enumerate}
\section{Optimization Framework} \label{sec:dr-mlsp}
In this section, we characterize the discrete rate minimum length scheduling problem referred as $\cal{DR-MLSP}$. We first present the mathematical formulation of $\cal{DR-MLSP}$ as an mixed integer non-linear optimization problem. Then, we prove the NP-hardness of $\cal{DR-MLSP}$. Finally, we propose the solution strategy followed in the subsequent sections.

\subsection{Mathematical Formulation}
The joint optimization of the time allocation, power control, rate adaptation and scheduling with the objective of minimizing the schedule length given the traffic demands of the users while considering realistic non linear energy harvesting model for an in-band full-duplex system is formulated as follows:

$\cal{DR-MLSP}$:
\begin{subequations} \label{opt_problem}
\begin{align}
& \textit{minimize}
& & \sum_{i=0}^{N}\tau_i \label{obj}\\
& \textit{subject to}
& & P_ig_i-\bigg(\sum_{k=1}^{M}z_{ik}\gamma^k \bigg)\bigg(\sigma_o^2+\beta P_h\bigg)\geq0,  \label{datarate}\\
&&& B_i+C_i\tau_0+C_i\sum_{j=1,i\neq j}^{N}a_{ij}\tau_j+C_i\tau_i-P_i\tau_i\geq 0 \label{energy_causality}\\
&&& \tau_i\chi_i\geq D_i,  \label{trafficDemand}\\
&&& P_i\leq P_{max}. \label{pmax} \\
&&& a_{ij}+a_{ji}=1,i\neq j \label{ordering}\\
&&& \chi_i=\sum_{k=1}^{M}z_{ik}r^k \label{rate}\\
&&& \sum_{k=1}^{M}z_{ik}=1 \label{OneRate}\\
& \textit{variables}
& & P_i \geq 0, \hspace*{0.1cm} \tau_i\geq 0, \hspace*{0.1cm} a_{ij} \in \{0,1\}, \hspace*{0.1cm} z_{ik} \in \{0,1\}.\label{pcp1_vars}
\end{align}
\end{subequations}

The variables of the problem are $P_i$, the transmit power of user $i$; $\tau_i$, the transmission time of user $i$, $a_{ij}$, binary variable that takes value $1$ if user $i$ is scheduled before user $j$ and $0$ otherwise and $z_{ik}$, a binary variable which takes a value $1$ if user $i$ is allocated rate $r_k$ and $0$ otherwise. In addition, $\tau_0$ denotes an initial waiting time duration during which all the users only harvest energy without transmitting any information.

The objective of the optimization problem is to minimize the schedule length as given by Equation (\ref{obj}). Equation (\ref{datarate}) represents the constraint on satisfying the rate adaptation of the users. Equation (\ref{energy_causality}) gives the energy causality constraint: The total amount of available energy, including both the initial energy and the energy harvested until and during the transmission of a user, should be greater than or equal to the energy consumed during its transmission. Equation (\ref{trafficDemand}) represents the traffic demand constraint of the users. Equation (\ref{ordering}) represents the scheduling constraint, stating that if user $i$ transmits before user $j$, user $j$ cannot transmit before user $i$. Equation (\ref{pmax}) represents the maximum transmit power constraint. Whereas, Equation (\ref{rate}) and Equation (\ref{OneRate}) represents the adapted rate and its uniqueness respectively.

This optimization problem is a Mixed Integer Non-Linear Programming problem thus difficult to solve for the global optimum \cite{opt_book}. 

\subsection{NP-hardness of the Problem} \label{sec:np}

\begin{theorem}
Discrete rate minimum length scheduling problem, $\cal{DR-MLSP}$, is NP-hard.
\end{theorem}

\begin{proof}
We prove by reduction from single machine makespan minimization problem with step improving times studied in \cite{cheng2006scheduling}. We first describe the make-span minimization problem: There are $n$ independent jobs $\lbrace J_1, J_2,...,J_n\rbrace$ available at time $t=0$ with a common critical date $D$. Each job $J_i$ has two processing times dependent on the starting time of the job. If job $J_i$ is started at $t<D$, then its processing time is $\alpha_i$; if it is started at $t\geq D$ otherwise, its processing time is $\beta_i$ where $0\leq \beta_i \leq \alpha_i$. The goal is to determine a non-preemptive schedule minimizing the make-span which is the completion time of the last scheduled job. 

We construct the following instance of $\cal{DR-MLSP}$: There are $n$ users, $M=2$ discrete data rates $r=(r^1, r^2)$ and SINR levels $\gamma=(\gamma^1, \gamma^2)$ where all users can afford to transmit with rate $r^1$ until a common time $T$ and with $r^2$ for $t\geq T$. Then it can be easily verified that the single machine make-span minimization problem with step improving times is polynomial-time solvable if and only if this particular instance of $\cal{DR-MLSP}$ is polynomial-time solvable. As a consequence, since the former problem is NP-hard, $\cal{DR-MLSP}$ is also NP-hard.
\end{proof}

\subsection{Solution Framework}

As the NP-hardness of the problem indicates, it is generally \textit{difficult} to solve $\cal{DR-MLSP}$ for a global optimum. In other words, finding a global optimum requires algorithms with exponential complexity. Such optimal algorithms, on the other hand, are intractable even for moderate problem sizes. In this paper, we present a solution framework to overcome this intractability based on the decomposition of the optimal power control and rate adaptation problem and the optimal scheduling problem as described below:

\begin{itemize}
\item For a given scheduling order of the users; i.e., for predetermined set of $a_{ij}$ values, $\cal{DR-MLSP}$ requires determining the optimal power control and rate adaptation of the users with minimum schedule length while considering their data, maximum transmit power and energy causality constraints. We first formulate this problem and introduce the minimum length scheduling (MLS) slot definition which gives the best transmission slot for a user to complete the transmission earliest. Then, we present the optimal power control and rate adaptation policy.
\item Determining the optimal power control and rate adaptation for a given scheduling order reduces $\cal{DR-MLSP}$ to determining the optimal scheduling order. Based on the introduced MLS slot definition, we first classify the scheduling problem to non-overlapping (i.e. in which the best transmission slots are distributed in a way that they do not overlap each other over time) and overlapping MLS slots scenarios (i.e. in which the MLS slots of different users are coming into collision with each other over time). For non-overlapping scenario, we propose an optimal scheduling algorithm. For overlapping scenario, we propose a heuristic scheduling algorithm based on the investigation of the effect of the overlapping slots on the optimal schedule.
\end{itemize}

\section{Power Control and Rate Adaptation Problem}\label{sec:power}
In this section, we will investigate the discrete rate optimal power control and rate adaptation problem referred as $\cal{DR-PCP}$ for a given schedule i.e. $a_{ij}$ are given in the $\cal{DR-MLSP}$ with the objective to allocate optimal power, transmission time, and rate adaptation. Without loss of generality, we assume that the user $i$ is transmitting its data in the $i^{th}$ transmission slot in the schedule. The reformulated optimization problem is given below:\\
$\cal{DR-PCP}$:
\begin{subequations} \label{opt_problem}
\begin{align}
& \textit{minimize}
& & \sum_{i=0}^{N}\tau_i \label{obj1}\\
& \textit{subject to}
& & P_ig_i-\bigg(\sum_{k=1}^{M}z_{ik}\gamma^k \bigg)\bigg(\sigma_o^2+\beta P_h\bigg)\geq0,  \label{datarate1}\\
&&& B_i+C_i\tau_0+C_i\sum_{j=1}^{i}\tau_j-P_i\tau_i\geq 0 \label{energy_causality1}\\
&&& \tau_i\chi_i\geq D_i,  \label{trafficDemand1}\\
&&& P_i\leq P_{max}. \label{pmax1} \\
&&& \chi_i=\sum_{k=1}^{M}z_{ik}r^k \label{rate1}\\
&&& \sum_{k=1}^{M}z_{ik}=1 \label{OneRate1}\\
& \textit{variables}
& & P_i \geq 0, \hspace*{0.1cm} \tau_i\geq 0, \hspace*{0.1cm} z_{ik} \in \{0,1\}.\label{pcp1_vars}
\end{align}
\end{subequations}

In order to solve $\cal{DR-PCP}$ optimally, we will first define the minimum length scheduling (MLS) slot and then present the optimal algorithm.

We start by investigating the optimal power control policy for a single user. Initially, a user may not be able to transmit with even the minimum possible transmission rate level $r^1$ since the initial available energy $B_i$ may not be able to support the corresponding transmit power that will satisfy the SNIR constraint given by Equation (\ref{transmission_rate}). Each transmission rate level $r^k$ requires certain amount of energy available to complete the required data transmission for the user. Let $t_i^k$ be the first time instant at which user $i$ can afford to use transmission rate $r^k$ via satisfying the SNIR constraint $\gamma_i\geq \gamma^k$ using the harvested energy. Note that $t_i^1\leq t_i^2\leq ...\leq t_i^M$. 

 \begin{figure}[t]
 \centering
\includegraphics[width= 0.5 \linewidth]{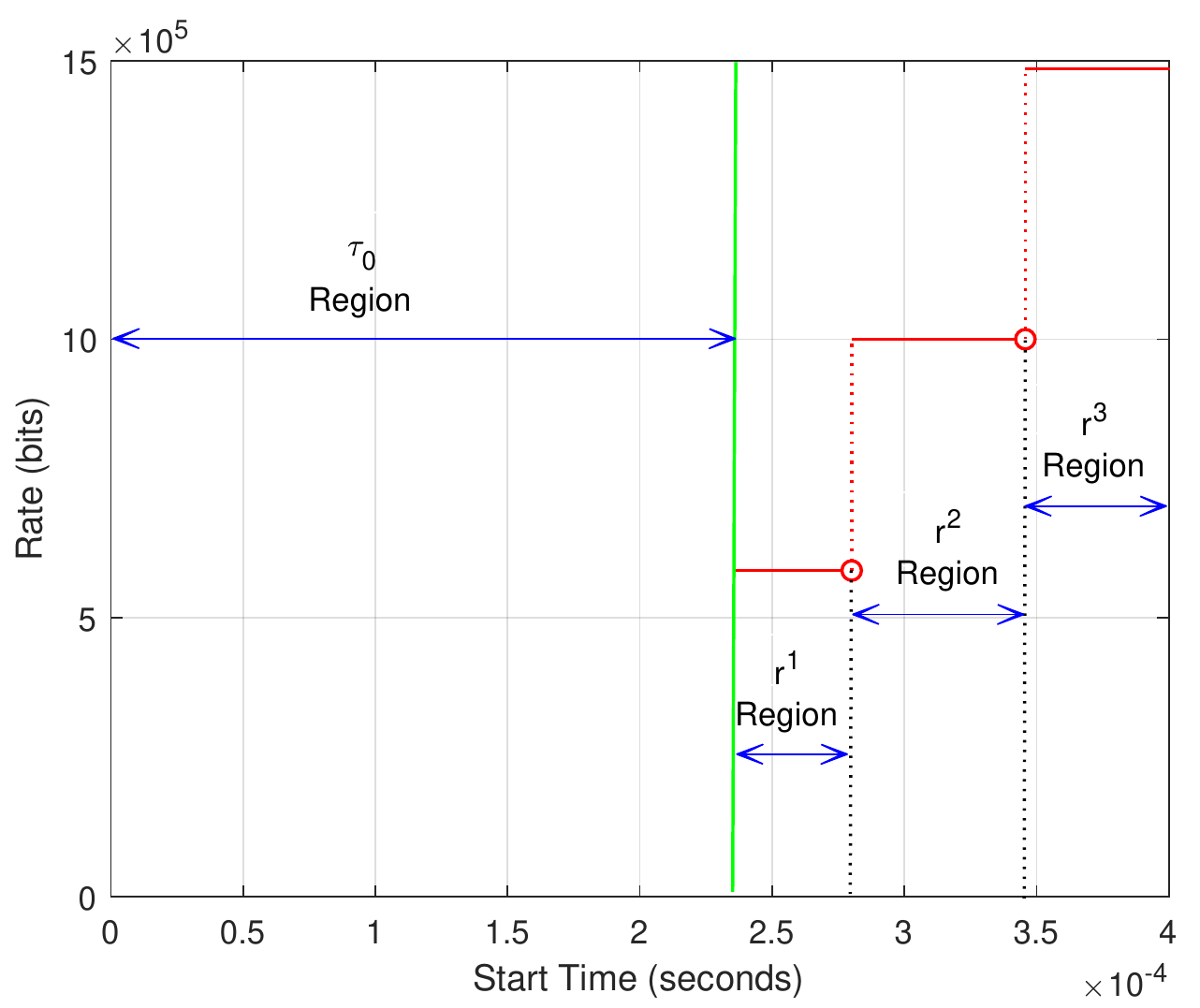}
\caption{Illustration of different rate regions for a user.} \label{fig:regions}
\end{figure}

 \begin{figure}[t]
 \centering
\includegraphics[width= 0.5 \linewidth]{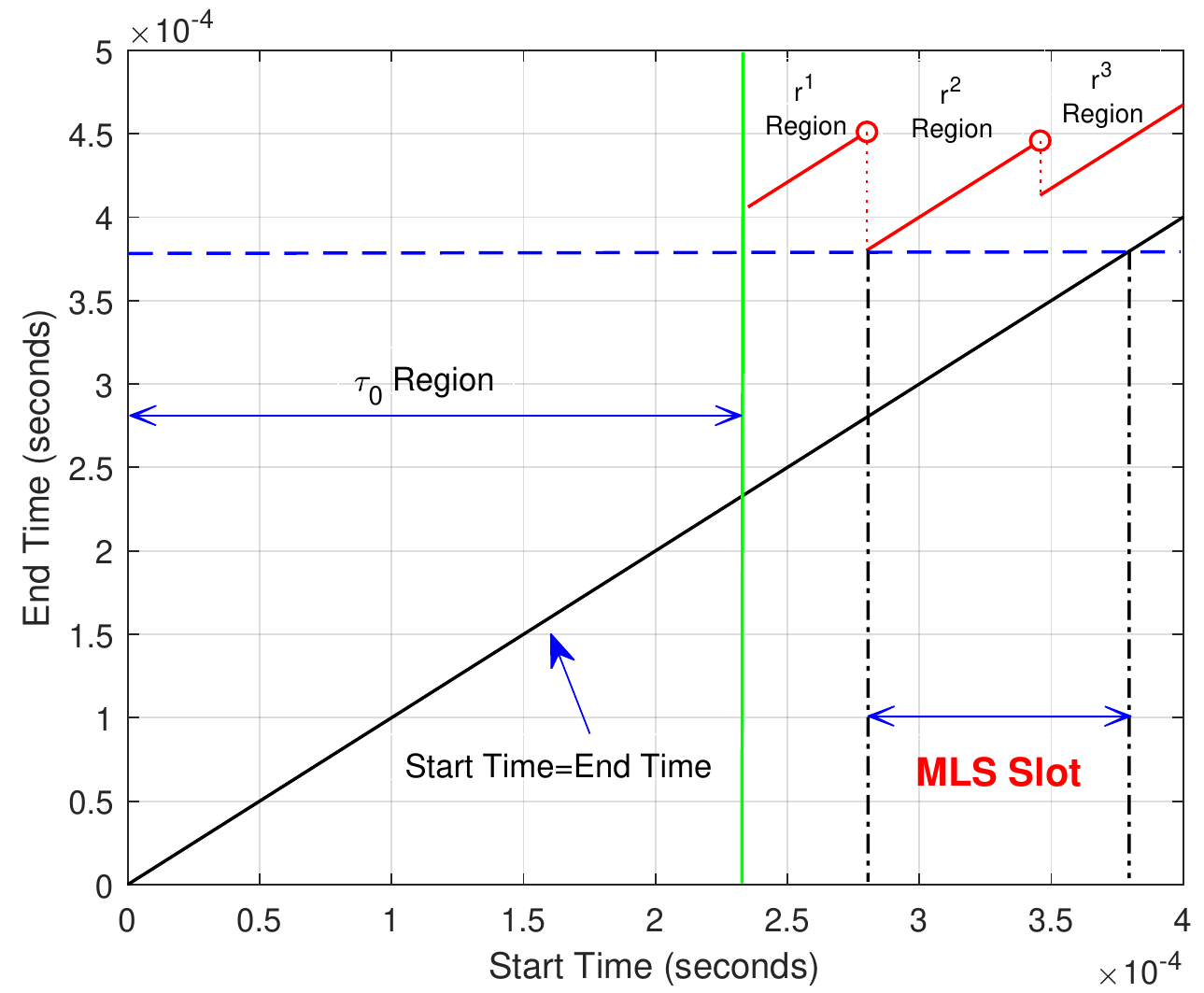}
\caption{Illustration of minimum length scheduling (MLS) slot for a user.} \label{fig:MLS_definition}
\end{figure}

Consider Fig.  \ref{fig:regions} illustrating the regions for transmission rates $r^k$s for a user for $M=3$ rate levels. $\tau_0$ region denotes the initial waiting time for user $i$ to be able to transmit with $r^1$ by satisfying $\gamma_i\geq \gamma^1$. Then, for a finite duration, user $i$ can only transmit with $r^1$; i.e., $r^{1}$ region. At $t_i^2$, user $i$ can support $r^2$ for the first time and $r^2$ is the maximum rate it can support for a specific duration; i.e., $r^2$ region. Then, finally at $t_i^3$, user $i$ can support $r^3$ for the first time and $r^3$ is the maximum rate it can support after $t\geq t_i^3$. Now, consider Fig. \ref{fig:MLS_definition} illustrating the transmission completion (end time) vs. allocation time (start time) for a user $i$. Let $s_i$ be the start time of the transmission of user $i$. Then, the completion time of the transmission is given by $e_i = s_i + \tau_i$ where $\tau_i$ is the transmission time of user $i$ such that $\tau_i=D_i/ \chi_i$. Since, user $i$ can support a higher transmission rate level at $t_i^k$ instants, $e_i$ is not a monotonically increasing function of $s_i$. At any $s_i=t_i^k$ value, $e_i$ decreases discontinuously by an amount of $D_i/r^{k-1}-D_i/r^{k}$ since the transmission rate jumps from $r^{k-1}$ to $r^{k}$. This suggests that waiting for the next transmission rate may decrease the completion time for a single user depending on the time instant at which the scheduling decision is given. 

Next, we introduce the minimum length scheduling (MLS) slot definition in the following.

\begin{definition} \label{def:MLS}
Let $t_i^{dec}$ is the time instant at which a scheduling decision is made for a user $i$. Then, let $s_i^*\geq t_i^{dec}$ be the starting time for user $i$ yielding the minimum completion time $e_i^*$ such that $e_i^*=\min_{s_i\geq t_i^{dec}} s_i+\tau_i(s_i)$. Minimum length scheduling (MLS) slot for user $i$ at time $t_i^{dec}$ is then defined as the time slot for user $i$ allocated in the interval $[s_i^*, e_i^*]$ where $s_i^*\geq t_i^{dec}$.
\end{definition}  

\begin{lemma} \label{lemma:MLSproperty}
For the MLS slot at time $t$ given by the interval $[s_i^*, e_i^*]$ , $e_i^*$ and $s_i^*$ are non-decreasing functions of $t$.
\end{lemma}
\begin{proof}
By Definition \ref{def:MLS}, MLS slot for a user $i$ at $t$ is defined as the interval $[s_i^*, e_i^*]$ where $e_i^*=\min_{s_i\geq t_i^{dec}} s_i+\tau_i(s_i)$. For two time instants $t_1$ and $t_2$ such that $e_i^1=t_1\leq t_2$,  $\min_{s_i\geq t_1} s_i+\tau_i(s_i) \leq \min_{s_i\geq t_2} s_i+\tau_i(s_i)=e_i^2$. Moreover, $s_i^1>s_i^2$ only if $e_i^1<e_i^2$ which is a contradiction. \end{proof}

Then, it is evident that for a single user, the DR-MLSP problem is solved by allocation of MLS slot at $t_i^{dec}=0$. Fig. \ref{fig:MLS_definition} illustrates the MLS slot for a single user $i$. Note that the MLS slot for user $i$ starts at $t_i^2$ where user $i$ can afford rate $r^2$ for the first time. This suggests that even if user $i$ can transmit with $r^1$ previously, the optimal policy is to wait until the time instant $t_i^2$ since the decrease in the transmission time due to this rate increase is larger than the waiting duration. Note that MLS slot for the user starts at a rate change instant $t_i^k$ which is not a coincidence for this specific scenario. The following lemma illustrates this behaviour.

\begin{lemma} \label{lemma:MLS}
MLS slot for user $i$ at $t_i^{dec}$ starts at either $t_i^{dec}$ or  $t_i^k\geq t_i^{dec}$ for some $k\in[1,M]$.
\end{lemma}

\begin{proof}
Let starting time for the MLS slot of user $i$ be $s_i^*$ such that $s_i^*\neq t_i^{dec}$ or $s_i^*\neq t_i^k$ for $k\in[1,M]$. Then, let the completion time be $e_i^*$. Suppose that $s_i^*$ is inside $r^l$ region; i.e., $t_i^{l} <s_i^* < t_i^{l+1}$. If $t_i^{dec}\leq t_i^{l}$, then starting the transmission time of user $i$ at $t_i^{l}$ decreases the completion time by $s_i^*-t_i^{l}$ since the transmission rate remains the same within $r^l$ region. If $t_i^{dec}> t_i^{l}$, then starting the transmission time of user $i$ at $t_i^{dec}$ decreases the completion time by $s_i^*-t_i^{dec}$. This is a contradiction by definition of MLS slot.
\end{proof}


Lemma \ref{lemma:MLS} illustrates that for a single user, the optimal scheduling policy is allocation of the user at either the time instant where the scheduling decision is made or at one of the time instants where the transmission rate changes for the user. Then, MLS slot determination can be made by evaluating at most $M+1$ time instants where $M$ is the number of rate levels. 

\begin{theorem} \label{thm:OPCA}
For a predetermined transmission order of users $\lbrace1, 2, ..., N\rbrace$, the optimal schedule is obtained by the allocation of MLS slot for each user $i$ at $t_i=e_{i-1}$ with $e_{0}=0$.
\end{theorem}
\begin{proof}
By Lemma \ref{lemma:MLSproperty}, the completion time of the MLS slot of a user is a non-decreasing function of the decision time. Then, in order to minimize the schedule length; i.e., the completion time of the last user, the decision for last user should be given at earliest which is the completion time of the previously scheduled user. Then, $t_N=e_{N-1}$. By the same logic, it can easily be verified that $t_i=e_{i-1}$ for $i\geq 2$. For user $1$, the earliest scheduling decision is $t=0$. This completes the proof.
\end{proof}

\begin{algorithm} []
\caption{Optimal Power Control Algorithm (OPCA)}  \label{algo_PDO}
\begin{algorithmic}[1] 
\STATE \textbf{input:} $\cal{F}$
\STATE \textbf{output:} $t(\cal{S})$
\STATE  $t^{dec}$ $\leftarrow$ 0, $t(\cal{S})$ $\leftarrow$ 0, $i$ $\leftarrow$ 1
\FOR {$i=1: |(\cal{F})|$}
\STATE determine MLS slots for user $i\in \cal{F}$ at $t^{dec}$,
\STATE $t^{dec}$ $\leftarrow$ $e_i^*(t^{dec})$,
\ENDFOR
\STATE $t(\cal{S})$ $\leftarrow$ $t^{dec}$,
\end{algorithmic}
\end{algorithm}

Based on Theorem \ref{thm:OPCA}, we present the Optimal Power Control Algorithm (OPCA), given in Algorithm  \ref{algo_PDO}. The algorithm is described as follows. Input of OPCA algorithm is a set of users, denoted by $\cal{F}$, with the characteristics specified in Section \ref{sec:system} (Line $1$). The algorithm starts by initializing the scheduling decision time  $t^{dec}$ (Line $3$). At each step of the algorithm, OPCA determines the MLS slots for the $i^{th}$ user at $t^{dec}$ (Line $5$), and allocates the transmission rate which gives the minimum transmission end time (Line $6$). Then, this user is allocated to its MLS slot starting at $s_k^*$ and completed at $e_k^*$. Algorithm continues by updating the scheduling decision time (Line $6$) and giving scheduling decisions for the remaining users (Lines $4-7$). OPCA terminates when all users in $\cal{F}$ are allocated the transmission slots and outputs the  schedule length $t(\cal{S})$ (Line $8$).

\section{Optimal Scheduling}\label{sec:scheduling}
The aim of this section is to determine the optimal schedule; i.e., the transmission order of the users, in order to minimize the length of the schedule. In Section \ref{sec:power}, we have determined the optimal time allocation and power control for a given transmission order of the users. On the other hand, optimizing the schedule can further decrease the schedule length. For instance, prioritizing the users that can reach their MLS slots earlier and delaying the users with further MLS slots may result in a significant decrease in the completion time of the transmissions. A straightforward solution to find the optimal schedule would be an exhaustive search algorithm that enumerates all possible transmission orders of the users and then determines the one with the minimum length using OPCA algorithm to determine the length of each. However, such an algorithm is evidently computationally intractable for even a medium size network due to $N!$ possible transmission orders for a network of $N$ users. However, computationally efficient solutions are required to provide scalability for different networking scenarios. 

We have shown in Section \ref{sec:np} that $\cal{DR-MLSP}$ is NP-hard. While this proves the non-availability of computationally efficient optimal algorithms, it does not necessarily mean that all instances of $\cal{DR-MLSP}$ require exponential-time solutions. In this section, we classify the $\cal{DR-MLSP}$ problem based on the distribution of the MLS slots of users over time into non-overlapping and overlapping scenarios. For the non-overlapping scenario, we propose optimal scheduling algorithm. For the overlapping scenario, we propose a polynomial-time heuristic scheduling algorithm.

\subsection{Non-overlapping MLS Slots Scenario}

In this section, we investigate $\cal{DR-MLSP}$ problem for the non-overlapping MLS slots (non-overlapping) scenario. First, we formally characterize the non-overlapping scenario. Let $[s_i^*, e_i^*]$ be the time interval representing the MLS slot of any user $i\in \cal{F}$ at time $t_i^{dec}=0$. Then, the condition for the non-overlapping scenario is given as follows:

\begin{equation}
s_j^* \not\in [s_i^*, e_i^*], \forall i,j \in \cal{F}
\end{equation}

 \begin{figure}[t]
 \centering
\includegraphics[width= 0.5 \linewidth]{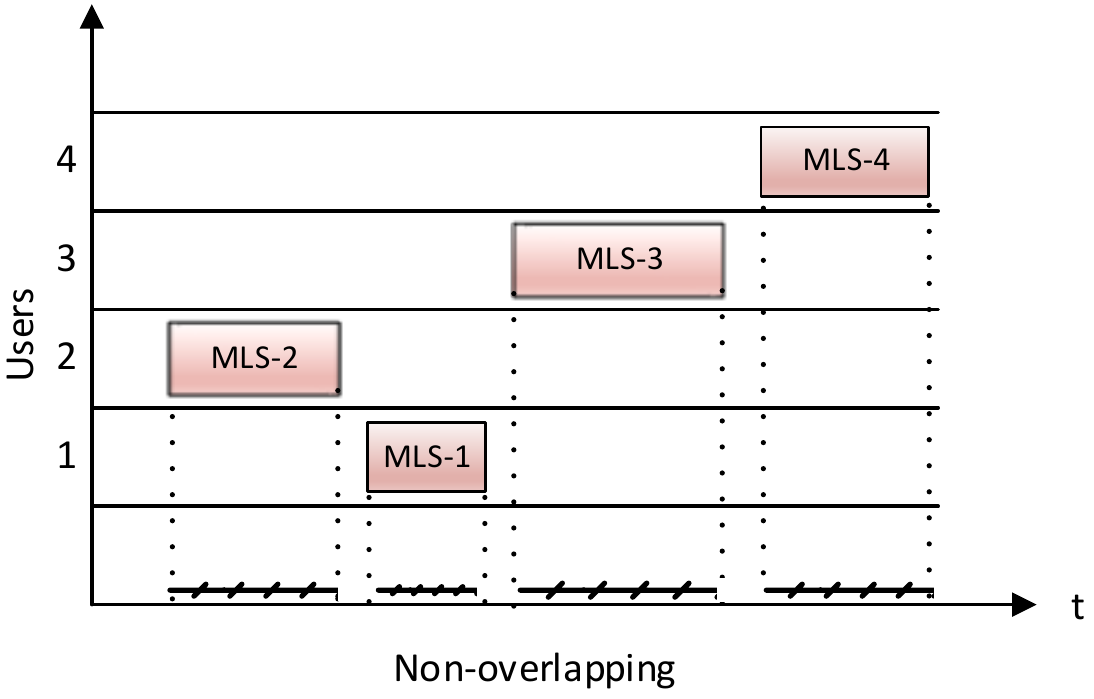}
\caption{Illustration of the non-overlapping MLS slots scenario.} \label{fig:nonoverlap}
\end{figure}

Non-overlapping MLS slots scenario is illustrated in Fig. \ref{fig:nonoverlap} for a network of $4$ users. As the figure illustrates, each user can transmit within its MLS slot without interfering with other users. Hence, the scheduling decision for any user can be given independently of the other users. Moreover, since MLS slot allocation is the optimal scheduling policy for a single user by definition, allocating their MLS slots to users in a non-overlapping scenario yields the optimal schedule. The optimal scheduling algorithm based on the allocation of the MLS slots, Optimal Scheduling for Non-overlapping Scenario (OSNS) is given in Algorithm  \ref{algo_OSNS}.

\begin{algorithm} []
\caption{Optimal Scheduling for Non-overlapping Scenario (OSNS)}  \label{algo_OSNS}
\begin{algorithmic}[1] 
\STATE \textbf{input:} $\cal{F}$
\STATE \textbf{output:} $t(\cal{S})$
\STATE  sort $s_i^*$ in increasing order s.t. $s_1^* < s_2^* < ... < s_{|\cal{F}|}^*$,
\FOR {$i=1: |(\cal{F})|$}
\STATE allocate MLS slot for user $i$ at $s_i^*$,
\ENDFOR
\STATE $t(\cal{S})$ $\leftarrow$ $e_{|\cal{F}|}^*$,
\end{algorithmic}
\end{algorithm}

As we have shown that this particular scenario covering particular instances of $\cal{DR-MLSP}$ is polynomial-time solvable, we now illustrate the probability of such instances in the following discussion.

\begin{figure}
\centering
\subcaptionbox{Probability of non-overlapping vs. number of users\label{fig:Pn}}
{\includegraphics[width=.45\linewidth]{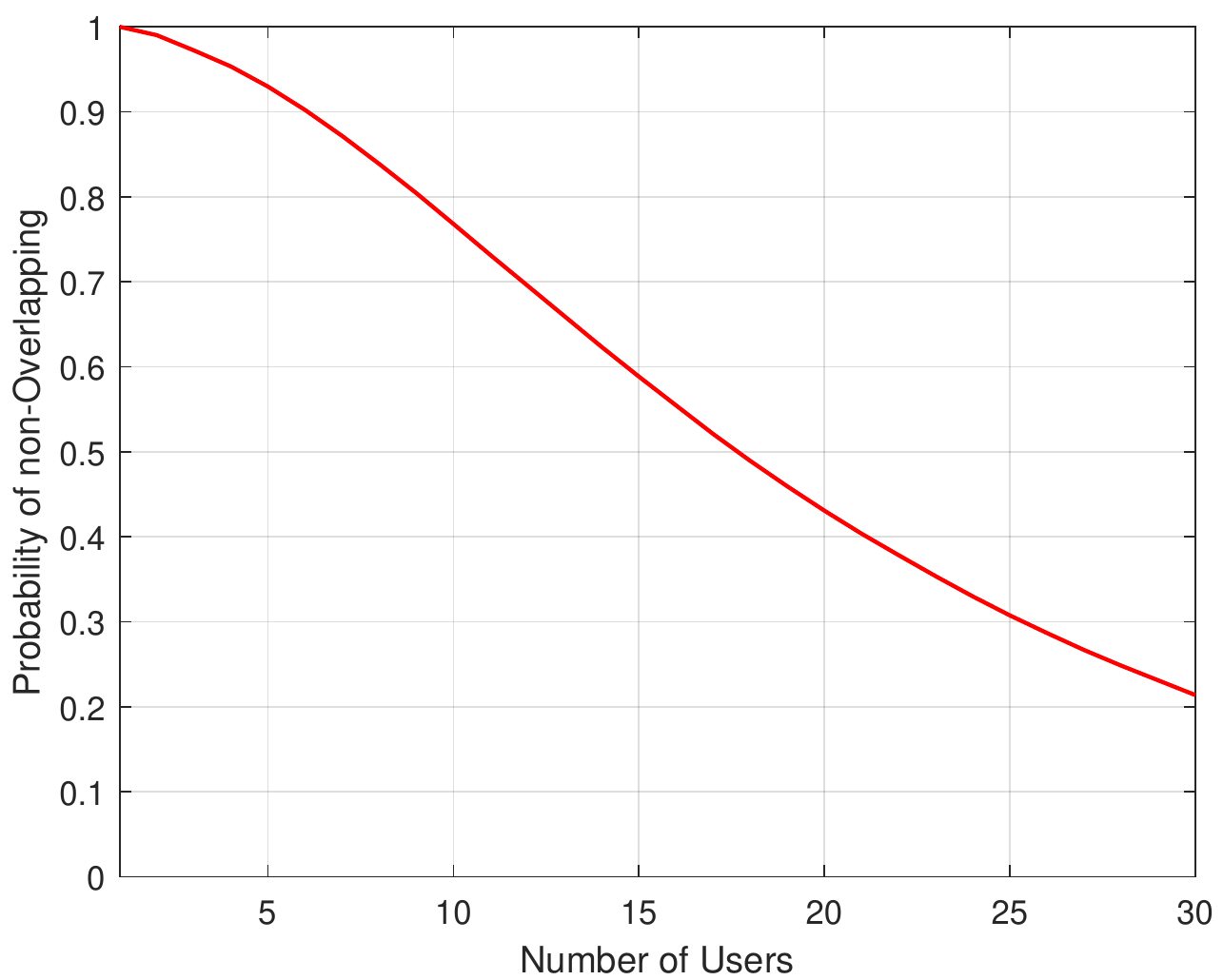}}
\subcaptionbox{Probability of non-overlapping vs. coverage radius \label{fig:Pdis}}
{\includegraphics[width=.45\linewidth]{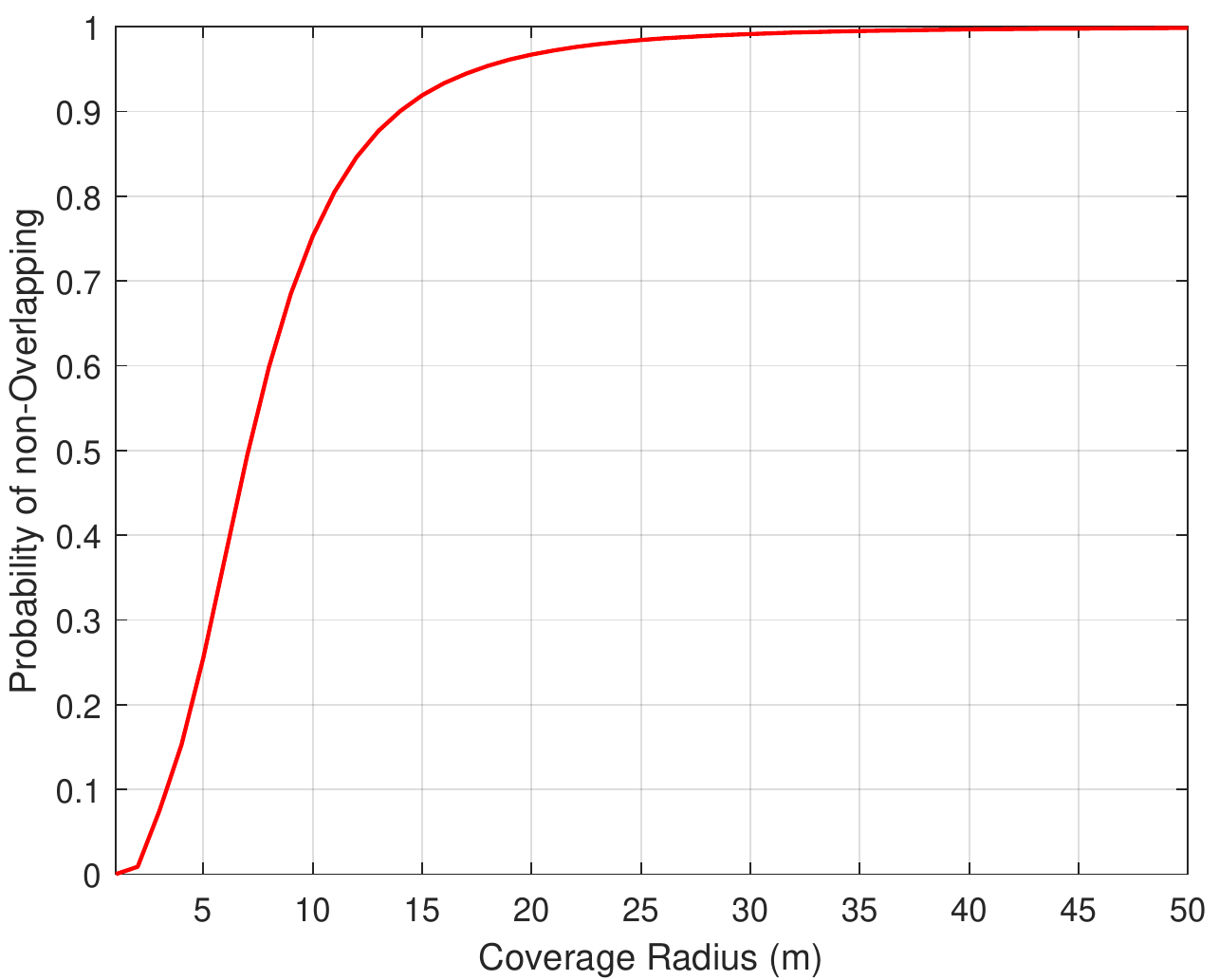}}
\subcaptionbox{Probability of non-overlapping vs. HAP transmit power $P_h$\label{fig:Phap}}
{\includegraphics[width=.45\linewidth]{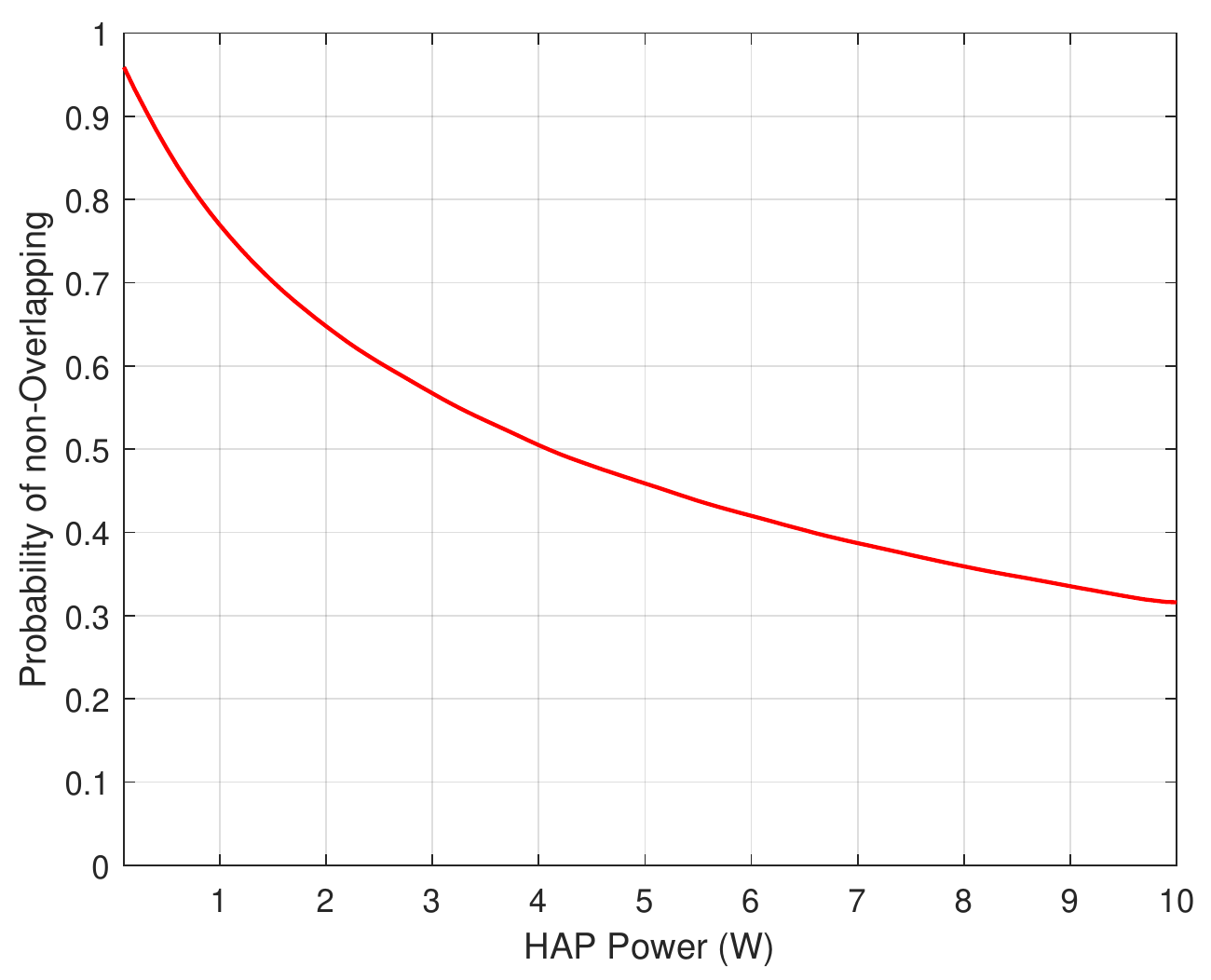}}
\subcaptionbox{Probability of non-overlapping vs. pathloss exponent\label{fig:Palpha}}
{\includegraphics[width=.45\linewidth]{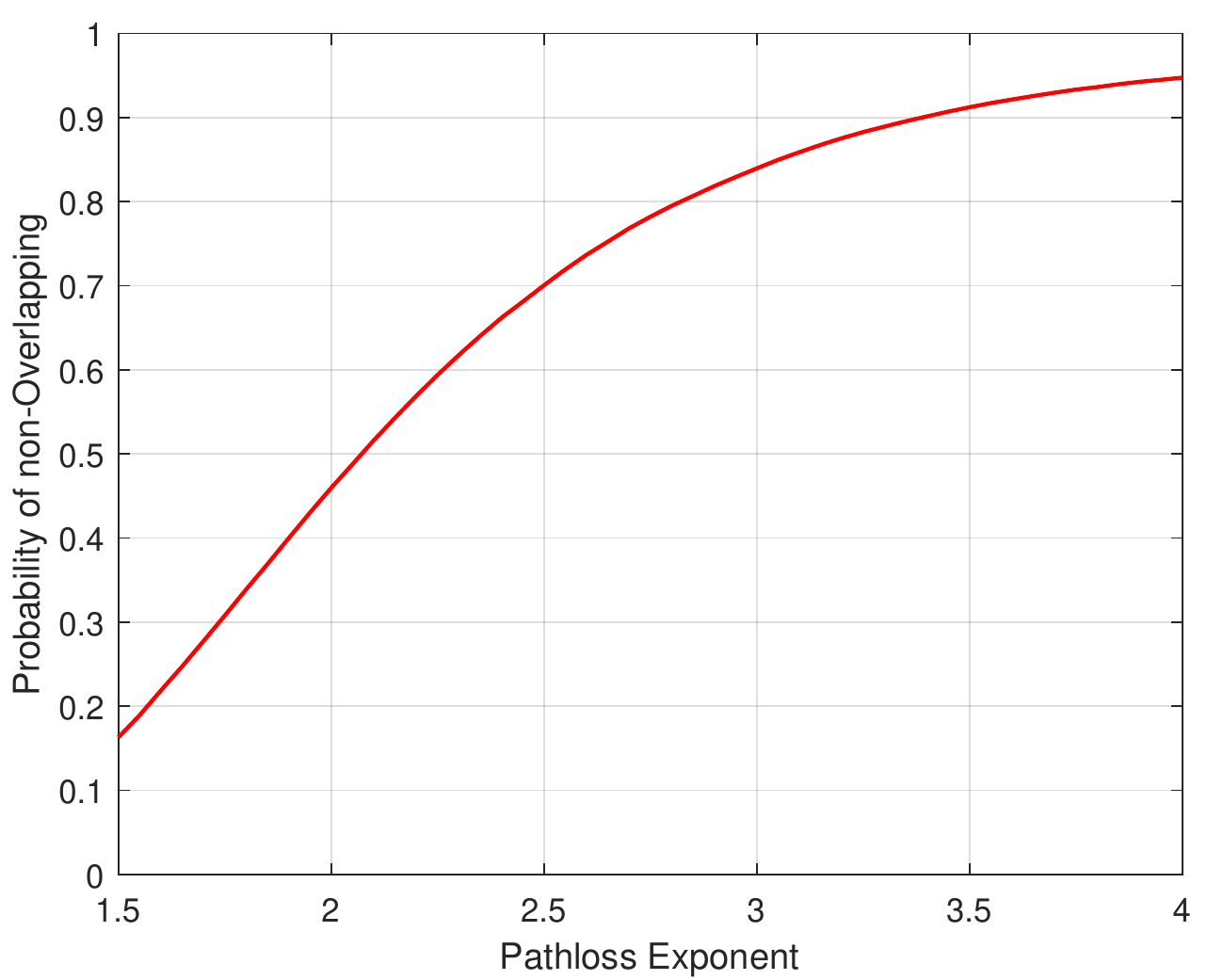}}
\caption{Probability of non-overlapping MLS occurrences}\label{fig:ProbAnalysis}
\end{figure}
%

Fig. \ref{fig:ProbAnalysis} illustrates the probability of having non-overlapping MLS slots in various network scenarios. Fig. \ref{fig:Pn} shows the behaviour of MLS slot when the network size increase. As for a single user there is no chance of overlapping, therefore probability of non-overlapping starts from one and decreases gradually as the number of users increase in the network because for more users there are more chances that two users may have MLS slot in the same region. \ref{fig:Pdis} shows that as the coverage of the network increases, the users can be more distant from the HAP and therefore, the channel gain reduces which results in a low energy harvesting rate. This low energy harvesting rate increases the energy difference among the users which results in a high probability of non-overlapping MLS slot. \ref{fig:Phap} shows that as the HAP power is increased the energy harvesting rate increases which is directly proportional to the HAP power. This increase allows the user to reach the highest SNR level quickly. Therefore, the probability of getting a non-overlapping MLS slots decreases. Similarly, \ref{fig:Palpha} shows that the high path-loss exponent will lead to high probability of getting non-overlapping slots due to low channel gains and low energy harvesting rates. Fig. \ref{fig:ProbAnalysis} that having a non-overlapping scenario, for which the optimal polynomial-time solution is given by OSNS Algorithm, is not a non-frequent scenario. For practical HAP power values, almost $75\%$ of the problem instances are polynomial-time solvable. While it is observed that the probability of having a non-overlapping instance decreases as the number of users increases, a large subset of instances fall in non-overlapping scenario for high path loss values (which can also correspond to fast variations in the channel) and for large coverage areas. Therefore, it is highly likely that a particular instance of $\cal{DR-MLSP}$ be polynomial-time solvable.

\subsection{Overlapping MLS Slots Scenario}
In this section, we investigate $\cal{DR-MLSP}$ problem for the overlapping MLS slots (overlapping) scenario. Overlapping MLS slots scenario is illustrated in Fig. \ref{fig:overlap} for a network of $4$ users. For instance, MLS slots of users $2$ and $1$ are overlapping meaning that at least one of the users cannot be allocated to its MLS slot.

\begin{figure}[t]
 \centering
\includegraphics[width= 0.5 \linewidth]{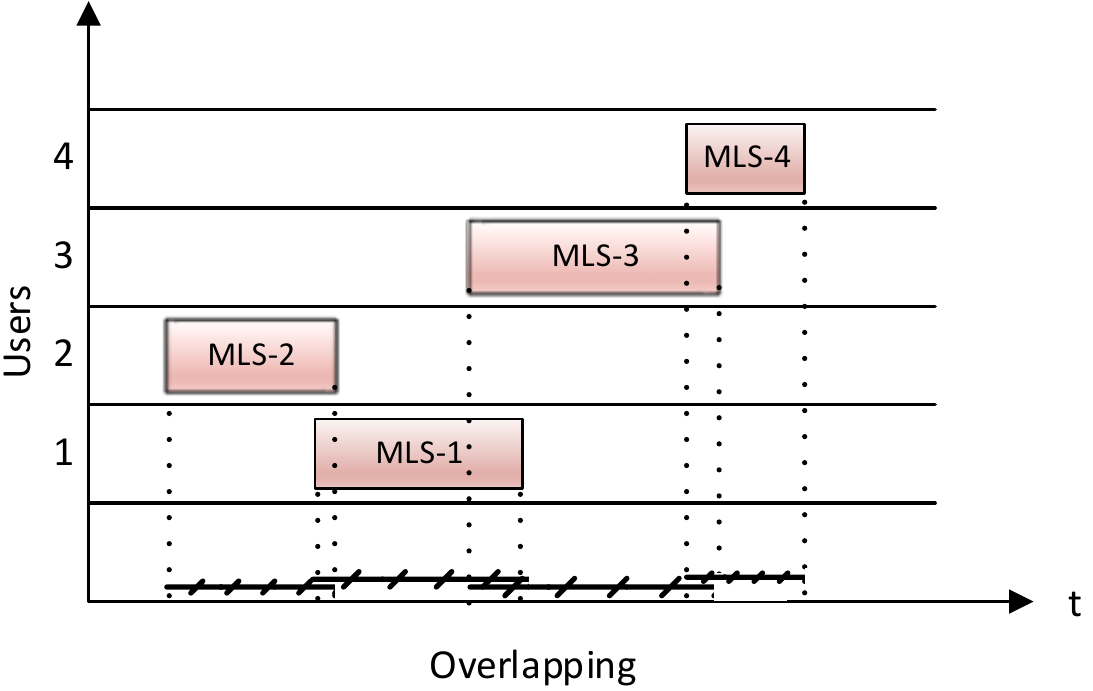}
\caption{Illustration of the overlapping MLS slots scenario.} \label{fig:overlap}
\end{figure}

While the overlapping MLS slots scenario is generally difficult to solve due to NP-hardness of the problem, some instances in this scenario may be still polynomial-time solvable using a polynomial-time conversion to a non-overlapping instance. In the following we propose an adjustment algorithm that determines the polynomial-time solvability of an overlapping MLS slots instance in polynomial-time.

The Polynomial-time Solvability Check Algorithm (PSCA), given in Algorithm \ref{algo_PSCA}, is described next. Input of PSCA algorithm is a set of users, denoted by $\cal{F}$. The algorithm starts by determining the MLS slots $[s_i^*,e_i^*]$ $\forall$ $i\in \cal{F}$ at $t=0$ (Line $3$). Then the users are sorted in increasing order of MLS slot starting times $s_i$ (Line $4$). At each step of the algorithm (Line $5$), PSCA determines whether subsequent MLS slots of two users are overlapping or not (Line $6$). If they are overlapping, the MLS slot of the later user is delayed such that its starting time is set to the completion time of the previous user (Lines $7-8$). If the MLS slot of the last user is not overlapping with the possibly adjusted time slot of the previous user (Line $11$), then set of users $\cal{F}$ is polynomial-time solvable (Line $12$).

\begin{algorithm} []
\caption{Polynomial-time Solvability Check Algorithm (PSCA)}  \label{algo_PSCA}
\begin{algorithmic}[1] 
\STATE \textbf{input:} $\cal{F}$
\STATE \textbf{output:} $[s_i^*,e_i^*]$ $\forall i\in \cal{F}$
\STATE determine MLS slots $[s_i^*,e_i^*]$ $\forall i\in \cal{F}$ at $t=0$,
\STATE  sort $s_i^*$ in increasing order s.t. $s_1^* < s_2^* < ... < s_{|\cal{F}|}^*$,
\FOR {$i=2: {|(\cal{F})|}-1$}
\IF {$s_i^*<e_{i-1}^*$}
\STATE set $s_i^*$ $\leftarrow$ $e_{i-1}^*$,
\STATE set $e_i^*$ $\leftarrow$ $e_{i-1}^*+\tau_i^*$,
\ENDIF
\ENDFOR
\IF {$s_{|(\cal{F})|}^*\geq e_{{|(\cal{F})|}-1}^*$}
\STATE $\cal{F}$ is polynomial-time solvable,
\ENDIF
\end{algorithmic}
\end{algorithm}

Note that if an overlapping instance of the problem is determined to be polynomial-time solvable by PSCA algorithm, then OSNS can be used to solve this instance optimally since the adjusted slots determined by PSCA constitutes a non-overlapping instance.







Non-overlapping instances and adjustable overlapping instances of $\cal{DR-MLSP}$ are polynomial-time solvable. However, scheduling the overlapping MLS slots optimally generally requires determining which slots will be delayed among the overlapping ones to minimize the overall schedule length. Note that the ending time of the last MLS slot among the users is the lower bound on the minimum schedule length. Hence, one needs to focus on allocating the users such that the user with the last MLS slot is delayed for a minimum duration. For example, if one can allocate the users in an order such that the last scheduled user is allocated to the MLS slot of that user at $t_{dec}=0$, then the resulting schedule is optimal. Hence, it is generally beneficial to allocate the user with the earliest current MLS slot at the decision of scheduling time. In the following, we propose a polynomial-time algorithm based on this idea.

The Earliest MLS Slot Algorithm (EMSA), given in Algorithm \ref{algo_EMSA}, is described next. Input of EMSA algorithm is a set of users, denoted by $\cal{F}$, with the characteristics specified in Section \ref{sec:system} (Line $1$). The algorithm starts by initializing the schedule $\cal{S}$ where the $i^{th}$ element of $\cal{S}$ is the index of the user scheduled in the $i^{th}$ time slot and the scheduling decision time  $t^{dec}$ (Line $3$). At each step of the algorithm, EMSA determines the MLS slots for the unallocated users at $t^{dec}$ (Line $5$), and
picks the user with the earliest MLS slot starting time $s_i^*$ (Line $6$). Then, this user is allocated to its MLS slot starting at $s_k^*$ and completed at $e_k^*$ (Lines $7-8$). Algorithm continues by updating the scheduling decision time (Line $9$) and giving scheduling decisions for the remaining users (Lines $4-10$). EMSA terminates when all users in $\cal{F}$ are scheduled and outputs the schedule $\cal{S}$ with schedule length $t(\cal{S})$ (Line $11$).

\begin{algorithm} [t]
\caption{Earliest MLS Slot Algorithm (EMSA)}  \label{algo_EMSA}
\begin{algorithmic}[1] 
\STATE \textbf{input:} $\cal{F}$
\STATE \textbf{output:} $\cal{S}$, $t(\cal{S})$
\STATE $\cal{S}$ $\leftarrow$ $\emptyset$, $t^{dec}$ $\leftarrow$ 0,
\WHILE {$\textbf{F} \neq \emptyset$}
\STATE determine MLS slots for all $i\in \cal{F}$ at $t^{dec}$,
\STATE $k \leftarrow$ arg$\min_{i\in \cal{F}} s_i^*(t^{dec})$,
\STATE $\cal{S}$ $\leftarrow$ $\cal{S}$ + $\lbrace k\rbrace$,
\STATE $\cal{F}$ $\leftarrow$ $\cal{F}$ - $\lbrace k \rbrace$,
\STATE $t^{dec}$ $\leftarrow$ $e_k^*(t^{dec})$,
\ENDWHILE
\STATE $t(\cal{S})$ $\leftarrow$ $t^{dec}$,
\end{algorithmic}
\end{algorithm}

\section{Performance Evaluation} \label{sec:simulation}
The goal of this section is to evaluate the performance of the proposed scheduling algorithm named as generic algorithm in comparison to the pre-determined transmission order denoted by PDO. The PDO aims at minimizing the schedule length for a given transmission order of the users by allocating each user as early as possible, without considering scheduling as given in \ref{algo_PDO}. The generic algorithm uses all the proposed scheduling algorithms to find the best possible schedule in polynomial time complexity. In the first step the algorithm evaluates the MLS slots for all the users followed by determining the scenario of the slots by checking either they are overlapping or non overlapping. For the non-overlapping slots, the algorithm will use the OSNS algorithm presented in Algorithm \ref{algo_OSNS} to find the optimal schedule and corresponding schedule length. However, if the slots are overlapping then the generic algorithm will go for polynomial-time solvability check by using the Algorithm \ref{algo_PSCA} to find that the slots can be adjusted optimally or not by using the PSCA given in Algorithm \ref{algo_PSCA}. For the overlapping slots if the slots can not be adjusted then the generic algorithm will use earliest MLS slot algorithm presented in Algorithm \ref{algo_EMSA}to determine the schedule. The optimal solution is obtained by a brute force algorithm, denoted by BFA, enumerating all the possible transmission orders and then picking the best schedule with minimum length. Considering the exponential complexity of BFA, we have executed it only for seven users i.e. $N=7$ to show the performance of the proposed algorithms compared to the optimal one.

\subsection{Simulation Setup}
Simulation results are obtained by averaging over $1000$ independent random network realizations. The attenuation of the links considering large-scale statistics are determined using the path loss model given by 
\begin{equation}
PL(d)=PL(d_0)+10\alpha log_{10}\bigg(\frac{d}{d_0}\bigg)+\emph{Z}
\end{equation}
where $PL(d)$ is the path loss at distance $d$ in $dB$, $d_0$ is the reference distance, $\alpha$ is the path loss exponent, and $Z$ is a zero mean Gaussian random variable with standard deviation $\sigma$. The small-scale fading has been modelled by using Rayleigh fading with scale parameter  $\Omega_i$ set to mean power level obtained from the large-scale path loss model. The parameters used in the simulations are $\eta_i=1$ for $i \in [1,N]$; $D_i$ is picked randomly between $[100 \hspace*{0.1cm} 10000]$ bits for $i \in [1,N]$; $W= 1$ MHz; $d_0=1$ m; $PL(d_0)=30$ dB; $\alpha=2.76$, $\sigma=4$ \cite{harvest_50}. The self interference coefficient $\beta$ is taken as $-80$ dBm. We use $M=5$ discrete rates with values $10kbps, 20kbps, \cdots , 50kbps$ and corresponding SNIR levels. For non-linear energy harvesting model, $P_s=7mW$, $A=1500$ and $B=.0022$ \cite{NLEH_parameter}.

 \begin{figure}[t]
\centering
\includegraphics[width= 0.5 \linewidth]{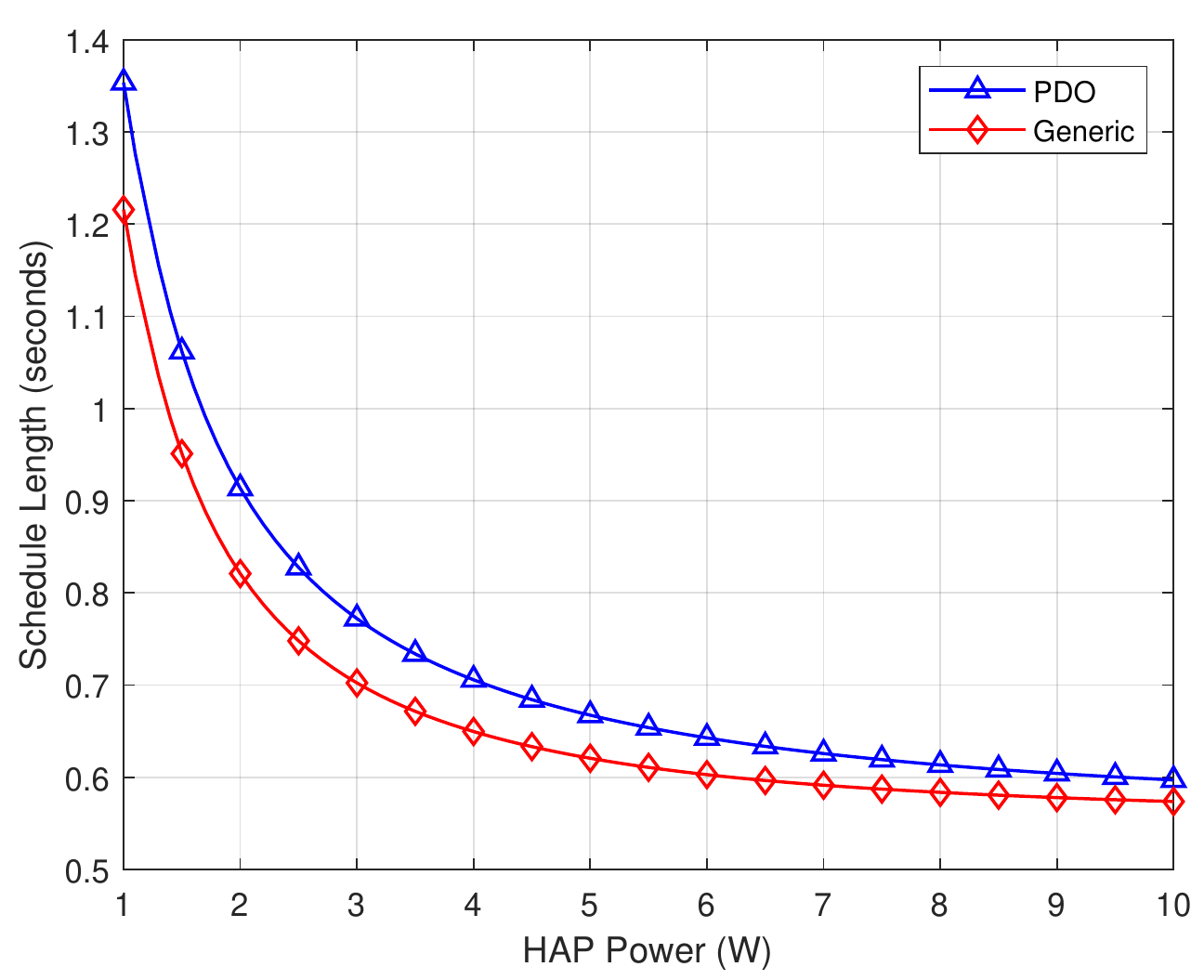}
\caption{Schedule length vs. HAP trasmit power $P_h$} \label{fig:Ph}
\end{figure}

\subsection{Scheduling Performance}

Fig. \ref{fig:Ph} illustrates the performance of the proposed generic algorithm for different values of the HAP transmit power for a network of 50 users assuming zero initial battery level. The schedule length is high for lower values of the $P_h$ and decreases as the $P_h$ increases because for high $P_h$ values, users can quickly achieve the higher SNR level resulting in faster data transmission and lower schedule length. The generic algorithm outperforms the PDO significantly for wide range of HAP power $P_h$ values. It is important to notice that for lower values of $P_h$, scheduling is more critical due to the fact that the users with low energy level will need more time to achieve the higher data rate. Therefore, if such users are scheduled in the start, they will be transmitting at lower rate which will result in higher transmission time. Therefore, delaying such users will help in schedule length minimization. The impact of the scheduling reduces as the $P_h$ is increased, because for high $P_h$ values all the users are harvesting energy at higher rate which facilitates them in achieving higher SNR values. If $P_h$ is increased enough so that all the users can feasible afford highest data rate then the schedule length for PDO and generic algorithm will converge to each other which removes the need of scheduling. Note that the lower bound on the schedule length is $\sum_{i=1}^{N}\tau_{i}^{min}$, where $\tau_{i}^{min}$ is the transmission time of user $i$ corresponding to the highest transmission rate i.e. $50kbps$ in this case.      
 \begin{figure}[t]
 \centering
\includegraphics[width= 0.5 \linewidth]{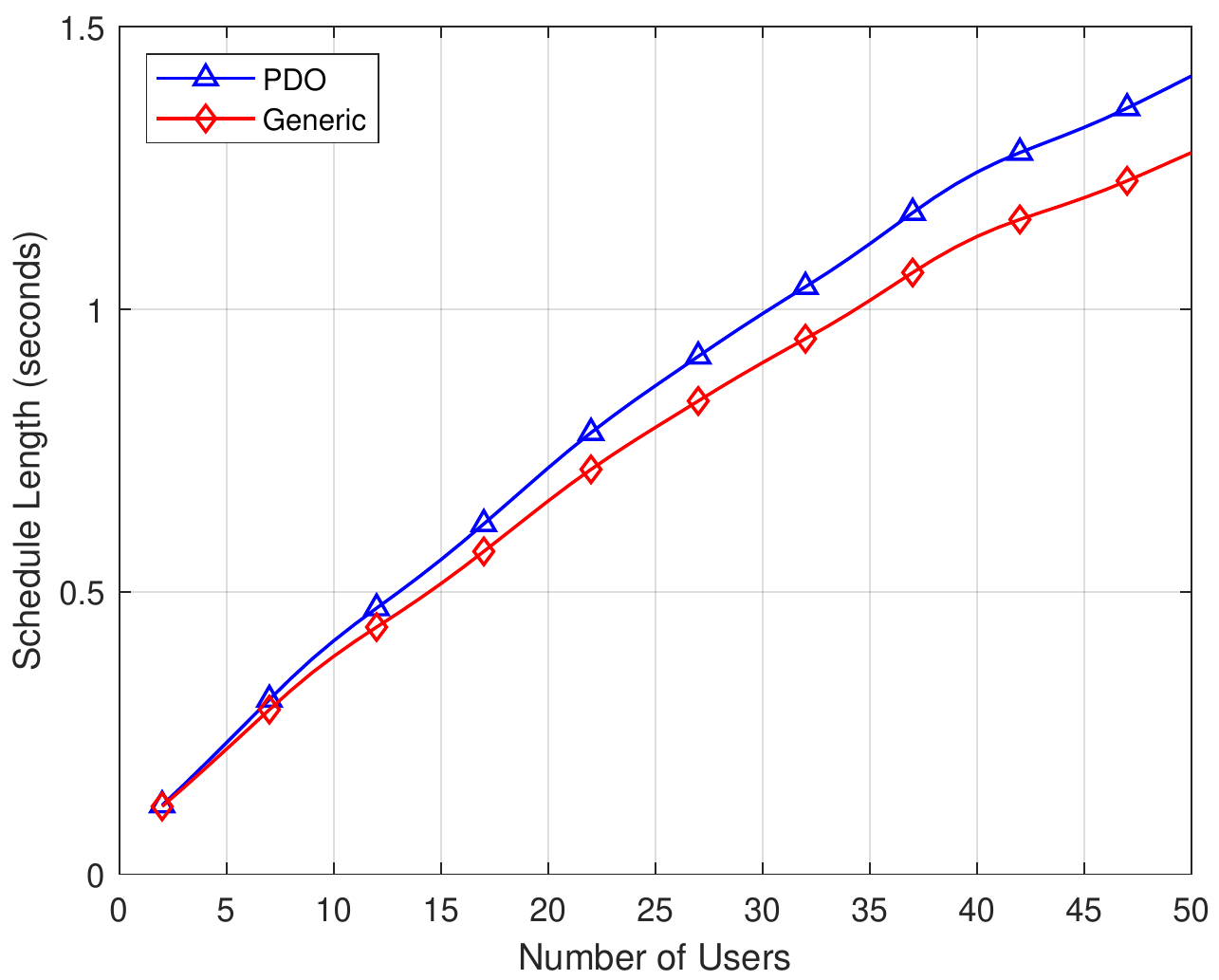}
\caption{Schedule length vs. number of users} \label{fig:N}
\end{figure}

Fig. \ref{fig:N} illustrates the impact of scheduling on the network size by comparing the PDO and generic algorithm for different number of users in the network for $P_h=1W$ and zero initial battery level. Note that due to exponential complexity of the brute force enumeration, we have skipped the BFA in this analysis and at the end we have presented it for lesser number of users. The figure illustrate that for a small network, the addition of each new user causes a linear increase in the schedule length. However as the network size increase the effect on the schedule length start diminishing because for a large network, the probability of finding a user with higher transmission rate increase for every slot and the users with low energy level will get more time to harvest the energy which will allow them as well to achieve high SNR level. Furthermore, as the number of users increase in the network, the scheduling becomes more critical because an arbitrary transmission order may result an unnecessary delay in the transmission of high rate users even if they do not need to harvest more energy. Moreover, increasing number of users puts an arbitrary transmission order further away from the optimality hence increasing the sub optimality of PDO. On the other hand, if the users are scheduled properly by eliminating the unnecessary waiting intervals for achieving higher rates as in the generic algorithm, the optimality performance is preserved showing the robustness of proposed algorithm to the network size. Note that robustness to the network size is very important for future networks with high number of machine type devices or sensors.

 \begin{figure}[t]
 \centering
\includegraphics[width= 0.5 \linewidth]{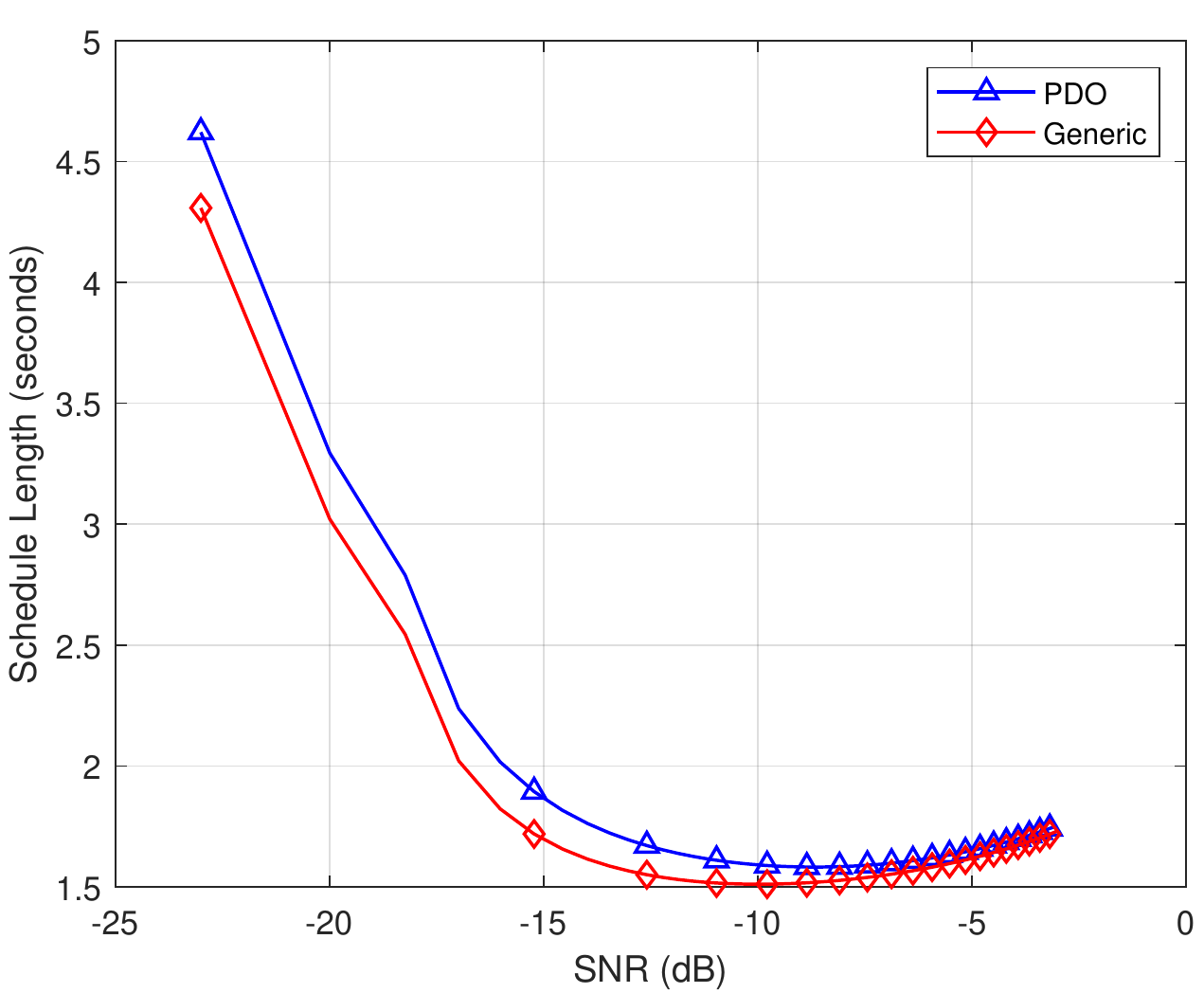}
\caption{Schedule length vs SNR values} \label{fig:SNR}
\end{figure}

Fig. \ref{fig:SNR} illustrates the behaviour of the network for different required SNR levels for minimum rate transmissions. This system can be assumed as constant rate model in which all the users are supposed to transmit their information by using a fix SNR level. It is observed that initially, increasing the minimum SNR level results in reduction of the schedule length. This reduction is due to the fact that the users may harvest more then the needed energy even during the transmission of first user and later on, all the users are just waiting for their transmission time. Due to low value of the allowed SNR, the users will be under-utilizing their energy. Therefore, increase in the minimum allowed data rate will allow the remaining users to fully utilize their harvested energy which will lead to a reduction in schedule length. However, after a certain time, the schedule length starts increase, this increase in the schedule length is due to the longer waiting time of the users to achieve the required SNR level.

 \begin{figure}[t]
 \centering
\includegraphics[width= 0.5 \linewidth]{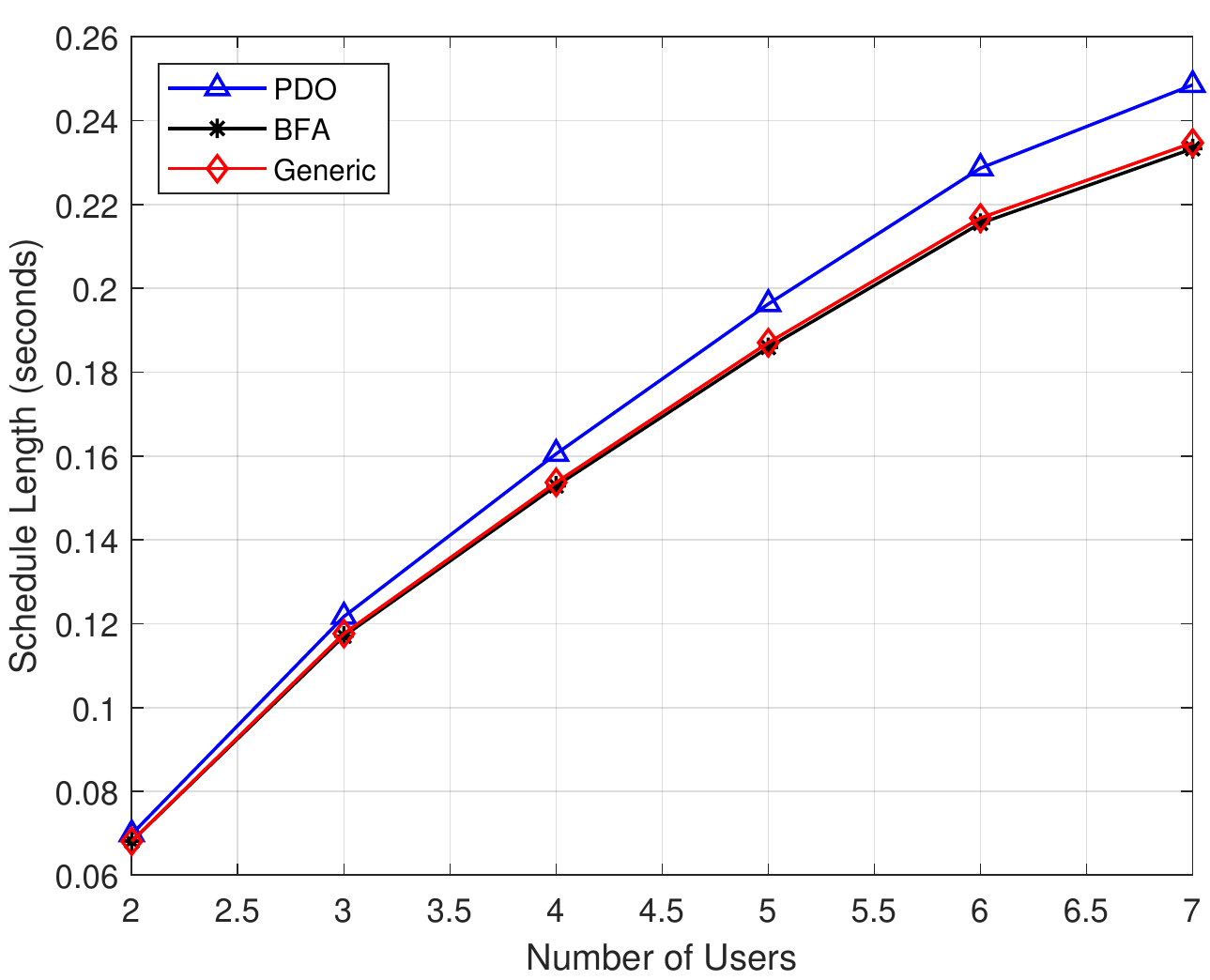}
\caption{Schedule length vs SNR values} \label{fig:SNR}
\end{figure}
Finally, Fig. \ref{fig:SNR} illustrates the comparison of the proposed algorithms with the optimal solution obtained by enumerating all the possible schedules and selecting the one with minimum length. As mentioned in the probabilistic analysis of the proposed problem, most scenarios of the problem are optimally polynomial time solvable except the few overlapping non-adjustable scenarios. Due to this reason the proposed generic algorithm performs very close to the optimal with a significant low complexity. For smaller number of users, the slots are non-overlapping with a high probability therefore the proposed algorithm is performing very close to the optimal. However,as the network size increases the proposed algorithm deviates a little from the optimal because for high number of users, the probability of overlapping non-adjustable scenario increases.  
\section{Conclusion and Future Work} \label{sec:conclusion}
In this paper, we have investigated minimum length scheduling problem considering discrete rate transmission model in a full duplex wireless powered communication network. We have characterized an optimization framework to determine the optimal time allocation, rate adaptation and scheduling subject to maximum transmit power, traffic demand and energy causality requirements of the users. First we have mathematically formulated the problem as a mixed integer non-linear programming problem which is difficult to solve for the global optimum in polynomial time. It is proved that the formulated problem is NP-Hard. In order to solve the problem fast and efficiently, the problem is analysed extensively and based on this analysis we have introduced the minimum length slot. All the network scenarios are divided into overlapping and non-overlapping scenarios. For non-overlapping scenarios, we have proposed optimal solution whereas, the overlapping scenarios are further analysed for adjustable and non-adjustable scenarios. For adjustable topologies, an optimal algorithm is proposed and as the non-adjustable scenarios can not be solved optimally in polynomial time, we have proposed a heuristic algorithm for such scenarios which perform very close to the optimal solution.   

\bibliography{bib_shahid}
\bibliographystyle{ieeetr}
\end{document}